%% file: main.tex
\documentclass{article}

\usepackage[english]{babel}

\usepackage[letterpaper,top=2cm,bottom=2cm,left=3cm,right=3cm,marginparwidth=1.75cm]{geometry}

\usepackage{amsmath}
\usepackage{amsthm} 
\usepackage{graphicx}
\usepackage[colorlinks=true, allcolors=blue]{hyperref}
\usepackage{cleveref}
\usepackage{xcolor}
\usepackage{todonotes}
\usepackage{thm-restate}
\usepackage{amssymb}

\newtheorem{theorem}{Theorem}
\newtheorem{corollary}[theorem]{Corollary}
\newtheorem{question}{Question}
\newtheorem{open}{Open problem}
\newtheorem{observation}[theorem]{Observation} 

\author{Laurent Feuilloley, Josef Erik Sedláček, Martin Slávik}

\title{Proving there is a leader without naming it}

\begin{document}
\maketitle

\begin{abstract}
Local certification is a mechanism for certifying to the nodes of a network that a certain property holds. In this framework, nodes are assigned labels, called certificates, which are supposed to prove that the property holds. The nodes then communicate with their neighbors to verify the correctness of these certificates.

Certifying that there is a unique leader in a network is one of the most classical problems in this setting. It is well-known that this can be done using certificates that encode node identifiers and distances in the graph. These require $O(\log n)$ and $O(\log D)$ bits respectively, where $n$ is the number of nodes and $D$ is the diameter. A matching lower bound is known in cycle graphs (where $n$ and $D$ are equal up to multiplicative constants).

A recent line of work has shown that network structure greatly influences local certification. For example, certifying that a network does not contain triangles takes $\Theta(n)$ bits in general graphs, but only $O(\log n)$ bits in graphs of bounded treewidth. This observation raises the question: Is it possible to achieve sublogarithmic leader certification in graph classes that do not contain cycle graphs? And since in that case we cannot write identifiers in a certificate, do we actually need identifiers at all in such topologies? We prove the following results.

\begin{itemize}
    \item For graphs with constant diameter, which is the most natural way to rule out cycle graphs in this context, $\Omega(\log n)$-bit certificates are still required.  
    \item For anonymous chordal graphs and anonymous grids, a leader can be certified using $O(\log D)$ bits. In particular, in chordal graphs with small enough diameter, we get a sublogarithmic certification. These two graph classes enforce in two different ways the fact that, unlike cycle graphs, the network has no large ``hole''.
    \item In graphs with logarithmic minimum degree, we obtain an $O(\log \log n)$ certification for leader. Here, identifiers are actually used, but they are not stored verbatim in the certificates. This indicates that sparsity also plays a role in the fact that cycle graphs are hard instances. 
\end{itemize}

We also discuss the types of properties to which these results apply beyond leader election, as well as the impact of the identifier range.
\end{abstract}


 \thispagestyle{empty}

\clearpage
\setcounter{page}{1}

\section{Introduction}

\paragraph{Local certification and certification size}
Local certification is a model of distributed computing, where the nodes of a network have to verify a property, with the help of an oracle. 
More precisely, each node of the network is first given a \emph{certificate} by a \emph{prover}, then it communicates with its neighbors, and finally outputs a binary decision \emph{accept} or \emph{reject}~\cite{Feuilloley21}. Such a scheme is correct for verifying a given property if the following holds: for any network configuration, there exists an assignment of certificates such that all nodes accept, if and only if, the configuration satisfies the property.
One can consider properties of the graph itself, such as having an odd number of nodes, or properties of an input labeling, for example \emph{having a unique leader} consists of having exactly one node with input label 1 (selected) and all other nodes with input label 0 (non-selected). 
 
The main measure of quality of a local certification is the size of the certificates used, which one wants to minimize. The certificates are simply bit strings used as labels, hence this size is a number of bits. 
Another aspect is the communication radius of the nodes. In the original model, \emph{proof-labeling schemes}~\cite{KormanKP10}, the radius is 1, which means that every node just sends and receives one message from each (direct) neighbor. 
A more general model, \emph{locally checkable proof} model~\cite{GoosS16}, allows for any constant radius.

\paragraph{Example: Certifying a leader}
Let us describe the classic way to certify that there exists a \emph{unique leader} in a graph.
This will serve as an example for the notions given above and will also be the key problem in this paper.

On a correct instance, the prover does the following:
\begin{enumerate}
    \item It writes the identifier of the selected node into the certificate of all nodes. 
    \item It chooses a spanning tree rooted at the selected node
    \item It writes into the certificate of every node the identifier of its parent in the tree, as well as its distance to the root (in the tree). (If the node is the root there is no pointer to the parent.)
\end{enumerate}

Every node can check the local consistency of this labeling by communicating with its neighbors. More precisely, after reading the certificates of its neighbors and its own it can check the following:
\begin{enumerate}
    \item The claimed identifier of the root is the same for itself and its neighbors. 
    \item If the node has label 1 (selected), then its identifier is the one appearing in the certificate, it has not been given a parent identifier, and its distance is 0. 
    \item Otherwise, it has label 0, has been given a parent identifier, which corresponds to a neighbor, and whose distance is one less than its own.  
\end{enumerate}.

It is clear that in correct instances the prover's strategy described above makes all nodes accept. Now, intuitively, if the instance is not correct, then either there are several leaders and this is detected thanks to the leader identifier of the certificates, or there is no leader and this is detected when checking the spanning tree, which cannot point to a selected node. 

On graphs with $n$ nodes and diameter $D$, this certification uses $O(\log n+\log D)$ bits. Indeed, the standard hypothesis is that identifiers are encoded on $O(\log n)$ bits and the distances can be encoded on $O(\log D)$ bits. If we do not insist on making the diameter appears explicitly in the complexity this boils down to $O(\log n)$, since $D\leq n$.

Having a leader is actually one of the many problems for which the optimal certification size is $\Theta(\log n)$, and it is an essential primitive for all these problems. Let us mention that the classic terminology in distributed computing is \emph{leader election}, but because there is no election here, we will refer to this problem as \emph{(having a) unique leader}.

\paragraph{How does the topology impact the hardness of certification?}

In the domain of local certification, there are important recent efforts to understand how the topology of the graph influences the hardness of certification. Or in other words, if we assume that the graph belongs to some structurally constrained graph class, does it make it easier to certify some property? And if yes, what properties?

In that direction, we have seen a lot of progress on the upper bound side, in particular in the \emph{compact certification} regimes, that is, the certifications of (at most) logarithmic size. These take the form of meta-theorems inspired by the meta-theorem in algorithmic graph theory~\cite{Kreutzer11, SiebertzV24}.
More precisely, they are theorems of the following form: In graph class X, certifying a property of type Y can be done with $O(\log n)$ bits, where X is typically defined by a bounded parameter (\emph{e.g.} treewidth), and Y is a logic capturing many classic properties (first order logic, or monadic second order logic)~\cite{FeuilloleyBP22, FraigniaudMRT24, FraigniaudM0RT23, Cook0M25}. 

The impact of such structural restrictions can sometimes be spectacular. For example, certifying that a graph has no triangles requires a polynomial number of bits in general graphs, but one of these meta-theorems (\cite{Cook0M25}) automatically implies that logarithmic certificates are enough in bounded treewidth graphs. 
This line of work has had impact beyond local certification, when the ideas have been transferred to the CONGEST model, see~\cite{FominFGMRT24, FominF0RT24, FominFG0RT25}. 

There is nothing similar for lower bounds. That is, the lower bounds are proved in a few very specific graphs, and if one wants to understand the hardness of verifying a property in a graph class that does not contain these special graphs, then there is simply no lower bound available.  
To our knowledge, there are three lower bound techniques for proving $\Omega(\log n)$ lower bounds. The best known lower bound is by Göös and Suomela~\cite{GoosS16} (generalizing a proof of~\cite{KormanK07}) but it is only proved for cycles.  
The paper \cite{FeuilloleyH18} also has a logarithmic lower bound that applies to different problems, but it is also only for cycles.
Finally, some bounds are obtained by reduction to 2-party communication complexity (see \emph{e.g.}~\cite{Censor-HillelPP20} ), but again the structure of the lower bound instances is very rigid, in the sense that it must encode a fixed partition into two well-structured parts.

This situation indicates that we might get new meta-theorems in the $o(\log n)$ regime for interesting properties, if only we can rule out cycles and structures coming from communication complexity reductions. 
In particular, the neighborhood diversity parameter~\cite{Lampis12} in promising since when it is bounded, the diameter is bounded too, which rules out the first obstacle, and also the structure of the graph must be very simple which rules out the second. 

\paragraph{The Göös-Suomela lower bound}

In this paper, we will in particular focus on the Göös-Suomela lower bound, and try to understand to which extent it can be generalized and what are its inherent limitations.\footnote{Throughout the paper, we will use \emph{Göös-Suomela lower bound} to refer to the $\Omega(\log n)$ lower bound in \cite{GoosS16}, but it should be noted that the same paper has several important lower bounds in the quadratic regime too.} 

In order to have a more informed discussion on the power and limitations of the Göös-Suomela lower bound, let us sketch the technique. Consider the task of verifying that there is a unique leader and assume that there exists a local certification with $o(\log n)$ bits for it. 
We consider a family of \emph{yes}-instances, which are cycle graphs with exactly one leader, with specific identifier assignments.\footnote{Throughout the paper, we will use the words \emph{cycle graph} to insist on the fact that the full graph considered is a cycle.}
From these, we can create a family of new instances by cutting several of these cycles and plugging them into a larger cycle. 
These are \emph{no}-instances because they have several leaders, but a counting argument shows that at least one of these must be accepted. 
This is a contradiction with the correctness of the certification scheme, and establishes the $\Omega(\log n)$ bound.

In the original paper~\cite{GoosS16}, the lower-bound technique is described as a general framework and applied to a few concrete problems.
Intuitively, the lower bound applies to properties where it is required that \emph{some pattern appears at most once} in the graph. 
Having a unique leader is such a property, since we want at most one leader. 
Actually, for this paper it is relevant to see unique leader as the conjunction of two problems: At-Most-One-Selected (AMOS)~\cite{FraigniaudKP13} and At-Least-One-Selected (ALOS)~\cite{FeuilloleyH18}. 
In this paper, for all results AMOS and unique leader are equivalent, and we use AMOS which better captures the intuition of the lower bound.

\paragraph{In which topologies does the Göös-Suomela technique apply?}

In the proof of the Göös-Suomela lower bound, the structure of the lower bound instances is very rigid, since they are cycle graphs only. 
It is a strength, in the sense that it shows that even in very restricted setting the lower bound holds, but it is also a weakness. 
Indeed, the proof does not apply to settings where we restrict ourselves to a graph class that does not contain infinitely many cycle graphs. In particular, unlike other lower bounds, it is not sufficient to have a cycle in the graph to apply the result. The main question of this paper is the following:

\begin{question}
    What are the graph structural constraints that allow a $\Omega(\log n)$ lower bound, and what are the ones that allow a $o(\log n)$ upper bound? 
\end{question}

Since identifiers are encoded on $\Theta(\log n)$ bits, this question is tightly connected to a second question.

\begin{question}\label{question:ID}
    What are the graph structural constraints that allow certification in anonymous graphs? 
\end{question}

In order to answer these questions, 
we will inspect the natural features of cycle graphs one after the other. 
A first aspect is that cycles are obviously not acyclic. The complexity of leader election in trees is known to be constant: it is sufficient to give every node its distance modulo three from the leader. The reason why this works is that it basically gives an orientation to each edge, and having a unique sink is easy to check locally in trees.
Note that in comparison with the classic scheme, in trees, not only can one avoid identifiers in the certificates, but also compress the distance to only three different labels.

Let us now explore other structural aspects of cycle graphs in more detail. 

\paragraph{The role of the diameter}

Another important aspect of cycles is their linear diameter, which is especially relevant for our problem since the diameter appears naturally in the bounds. 
A natural question is whether in low-diameter graphs the logarithmic lower bound still holds or not.
Our first result is to prove that it does.

\begin{restatable}{theorem}{ThmDiameter}
\label{thm:diameter}
    Certifying AMOS in a graph with bounded diameter requires certificates of size $\Omega(\log n)$.
\end{restatable} 

This theorem undermines the hopes of establishing a sublogarithmic meta-theorem based on neighborhood diversity (or similar parameters) that we mentioned before (and indeed it is not hard to see that Theorem~\ref{thm:diameter} extends to bounded neighborhood diversity).

The proof of this theorem revisits the proof of \cite{GoosS16}, with two essential modifications. Although this is the case in the original proof, one does not need the \emph{yes}-instances to be connections of long paths: we can replace them by paths of constant length with an arbitrary dense and large part added at one end.
Second, one can combine four of these small parts, hence keeping the diameter under control, when the original proof was just considering arbitrary combinations.

At an intuitive level, what is needed for the proof to work is to have some ``hole'' in the graph, and that around this hole there are at least two places that are ``thin''. See Figure~\ref{fig:hole-chordal-grid}.
However, good and general definitions of what a hole is and what it means for a region to be thin remain elusive. 
In the following, we prove that one can go below $\Omega(\log n)$ ruling out holes, for a specific definition, or ruling out sparse graphs. But we also show that these definitions do not capture all the relevant cases.

\begin{figure}[!h]
    \centering
    \begin{tabular}{ccc}
    \begin{minipage}{0.3\textwidth}
        \input{hole-thin}
    \end{minipage} 
         &
    \begin{minipage}{0.3\textwidth}
        \begin{center}
        \input{chordal}
         \end{center}
    \end{minipage}
    & 
    \begin{minipage}{0.3\textwidth}
        \input{grid}
    \end{minipage}
    \end{tabular}
    \caption{Left picture: illustration of the type of graphs for which the Göös-Suomela proof works, with a hole and two thin parts (the graph is the gray part). Middle picture: a chordal graph. Right picture: a 2-by-3 grid.}
    \label{fig:hole-chordal-grid}
\end{figure}
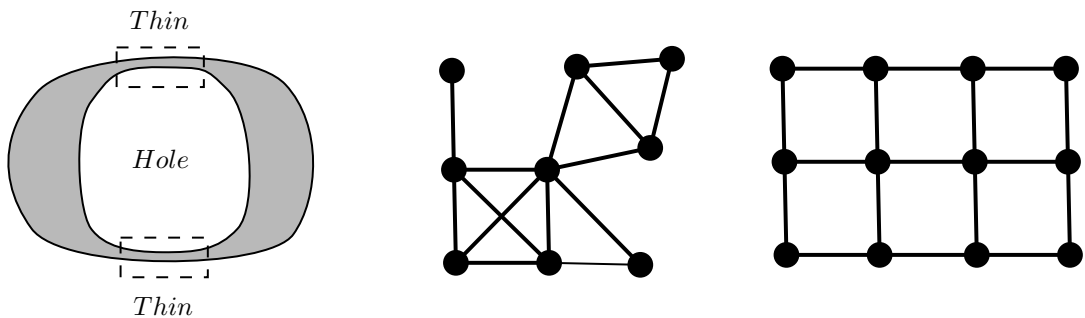

\paragraph{The role of the ``hole''}

A possible definition for a hole is a large enough induced cycle, that is, a set of nodes such that the edge set restricted to this set of nodes forms a cycle.
Graphs with no induced cycles of some length are well-studied in graph theory. 
The most classic such class (beside trees) is the class of \emph{chordal graphs}, which are the graphs that have no induced cycles of length larger than 3. 
These are basically the graphs that can be organized as trees of cliques, and they are a good example of a class for which there is no hole in the intuitive sense. See Figure~\ref{fig:hole-chordal-grid} for an example.

For chordal graphs, we show that there exists a local certification that only depends on the diameter, hence, being sublogarithmic if the diameter is small enough. And in the spirit of Question~\ref{question:ID}, we prove that we can do so in anonymous graphs.

\begin{restatable}{theorem}{ThmChordal}
\label{thm:chordal}
In anonymous graphs, there exists a $2$-local proof-labeling scheme $(f,\mathcal{A})$ of size $O(\log D)$ that accepts a chordal graph $G$ if and only if $G\in\mathcal{P}$, that is, if $G$ has at most one selected vertex.
\end{restatable} 

The certificates we use are very simple: they are only the distances to the selected node. 
The challenging part is to design the verification made at the nodes, and to prove correctness. 
For this the key part is to understand and use what the neighborhood of a node looks like in chordal graphs, which is also why we need radius-2, to be able to see these connections.

In the context of the paper, the take-away of this theorem is that defining holes as forbidden induced subgraphs leads to ``hole-free'' graph of small diameter breaking the lower bound, and that identifiers are not needed. 
But there are other cases where we can avoid identifiers, with a less conservative definition of a hole. An example is grid graphs (see Figure~\ref{fig:hole-chordal-grid}), where there are large induced cycles: consider for example the nodes at the border of the grid, they form an induced cycle of linear length. In these graphs we have a local certification, even in the anonymous setting. 

\begin{restatable}{theorem}{ThmGrid}\label{thm:grids}
    In anonymous graphs, there exists a 1-local proof-labeling scheme of size $O(\log D)$ that accepts a labeled grid graph of diameter at most $D$ if and only if at most one vertex is selected.
\end{restatable}

Note that here we only need radius 1, proving that the 2 of the previous result is not required to work in structured anonymous graphs. Also note that the logarithm of the diameter of the grid is of order $\Theta(\log n)$ so we do not break the lower bound in that case.

As mentioned earlier, the right definition of a hole eludes our understanding. Let us formulate this as an open problem (leaving aside the dense case that we discuss later).

\begin{open}
    In which (sparse) graphs does AMOS require $\Omega(\log n)$ bits? Is there a good notion of ``hole'', such that ruling out this structure captures the cases where we can go below $\Theta(\log n)$? 
\end{open}

An intuitive definition for a hole would be that it is a cycle in a graph, that would act as an obstacle in a road network: the paths that connect two diametrically opposite nodes of cycle can be partitioned in two groups, the one that ``go north of the obstacle'' and the ones that ``go south of the obstacle''. 
To our surprise, we were unable to find such a notion in the literature. Note that in grids and chordal graphs, there is no such hole, while the construction of Göös-Suomela  does have such a structure.   

\paragraph{The impact of sparsity/density}

Let us now discuss one more aspect of the cycle topology that makes Göös-Suomela proof work: the cycles are sparse graphs. 
In particular, the places where we cut and plug different instances are ``thin'', 
in the sense that there are not too many nodes close to these cuts. 
If this were not the case, the amount of information available in the certificates around the cut would not be negligible compared to the number of instances considered, and the counting argument would not work. 

Hence the question: What happens if the graph is dense everywhere? Is the sparsity requirement an artifact of the proof, or does density help? We show that density does help.

\begin{restatable}{theorem}{ThmDenseEverywhere}
     \label{thm:dense-everywhere}
    In graphs where every node has degree $\alpha \log^2 n$, for a large enough constant $\alpha$, AMOS can be certified with $O(\log \log n)$ bits with verification radius 2. 
\end{restatable}

Interestingly in the case of that theorem, even though the identifiers cannot be encoded in the certificates due to the lack of space, we do use them. The trick is that when we need to provide an identifier to a node, we can cut this identifier into many smaller pieces, and distribute the pieces to the neighbors of the node. The node can then collect these pieces from its neighbors' certificates. 

\paragraph{The role of the identifier range}

Our focus has been mostly to understand the scope of the Göös-Suomela technique in terms of topology. Another aspect of the technique to be questioned is the identifier range. Indeed, the construction of the instances leading to a contradiction requires the identifier range to be quadratic. Again, the question is: is this necessary, or is it an artifact of the proof? If the identifier range is of the form $[1, n+c]$ where $c$ is a constant, then the following theorem proves that once again the lower bound can be broken. 

\begin{restatable}{theorem}{ThmIDs}\label{thm:small_ids}
    Let $G$ be a graph with bounded diameter $D$. 
    If there exists a vertex whose identifier is of constant size relative to the size of $G$, then AMOS can be verified with certificates of size $O(\log D)$.
\end{restatable}

\paragraph{What are the relevant properties beyond AMOS and unique leader?}

We have mentioned before that meta-theorems usually have two restrictions: one on the topology of the graph and one on the set of problems considered. Our focus is on the former, but the second is also interesting. We have focused on AMOS and unique leader for simplicity, but our questions are relevant for any property in the logarithmic regime. 

For example, Göös and Suomela consider two other problems: verifying that a set of pointers forms a spanning tree and verifying that the graph has an odd number of nodes. It is not hard to argue that spanning tree reduces to AMOS, using the notion of local reduction recently introduced in~\cite{EsperetZ25}, but it is unclear whether this could be achieved for the odd number of nodes. This prompts the following question.\footnote{Some variant of this question appears in~\cite{Feuilloley21} as Open Problem 11.} 

\begin{open}
    Can we characterize the properties that have complexity $\Theta(\log n)$, at least on cycles? Is there a property that is complete for local reduction?
\end{open}

In this direction, let us note that a recent paper proved that in anonymous cycles without input labels, there is a gap in the complexity between $\Theta(1)$ and $\Theta(\log n)$, in the sense that no property can have optimal certificate size strictly in between these regimes~\cite{BousquetFZ25}. 
This means that in that restricted setting, the logarithmic regime forms a complexity class that is well-separated from the lower complexities.

\section{Model, definitions and notations}

All graphs in this paper are undirected, connected, simple graphs, with inputs labels on the nodes, denoted by $G = (V, E, L)$ where $L \colon V \to \{0,1\}^*$ encodes the labeling.
We denote by $n$ the number of vertices and by $D$ the diameter of the graph.
Unless specified otherwise, the vertices are assigned \emph{unique identifiers} encoded on $O(\log n)$ bits. We use the word \emph{anonymous} when there is no identifiers.
Neighbors of a vertex $v$ are denoted as $N_G(v)$, and if $G$ is clear from the context, $N(v)$ is used. Distance between two vertices $u, v$ is denoted as $d_G(u, v)$, and the subscript is omitted if $G$ is clear from the context.
We denote the set of vertices within distance $r$ from $v$ as $V[v,r]$, also called the $r$-local neighborhood of $v$.

A \emph{graph property} is formally a set of graphs that is closed under isomorphism, that is, its membership does not depend on the choice of identifiers.
Note that it may depend on the labels of the graph.
A \emph{certificate assignment} $P$ for $G$ is a function $P\colon V(G) \rightarrow \{0,1\}^*$ that associates with each vertex a \emph{certificate}. 
We say that $P$ has size $s$ if $|P(v)| \leq s(n)$ for every $v$.
A \emph{verifier} is a function that takes as an input a graph $G$, its certificate assignment $P$ and $v \in V(G)$ and outputs either $0$ or $1$.

The induced subgraph $G[V[v,r]]$ is denoted as $G[v,r]$ and the restriction of $P$ to $V[v,r]$ is denoted as $P[v,r]$, that is $P[v,r] \colon V[v,r]\to \{0,1\}^*$. 
A verifier $\mathcal{A}$ is $r$-\emph{local} if $\mathcal{A}(G,P,v) = \mathcal{A}(G[v,r],P[v,r],v)$ for all $G$, $P$, and $v$. 

An $r$-local \emph{proof labeling scheme} certifying a property of labeled graphs $\mathcal{P}$ is a pair $(f,\mathcal{A})$,
where $\mathcal{A}$ is an $r$-local verifier and $f$, called the \emph{prover}, assigns to each $G \in \mathcal{P}$ a certificate assignment $P$ such that the following properties hold.
\begin{itemize}
    \item \emph{Completeness}: If $G \in \mathcal{P}$, then $\mathcal{A}(G[v, r], P[v, r], v) = 1$ for all $v$, where $P =f(G)$.
    \item \emph{Soundness}: If $G \notin \mathcal{P}$, then for every certificate assignment $P'$, there is $v$ such that \\ $\mathcal{A}(G[v, r], P'[v, r], v)~=~0$.
\end{itemize}
We say that $(f, \mathcal{A})$ has size $s : \mathbb{N} \to \mathbb{N}$ if $|f(G)(v)| \leq s(|V(G)|)$ for all $G \in \mathcal{P}$ and all $v \in V(G)$.

A \emph{chordal graph} is a graph wit no induced cycles of length 4 or more. 

The $k\times q$ square grid graph is the cartesian product of a path on $k$ vertices and a path on $q$ vertices. A \emph{grid graph} is one of these graphs with arbitrary $k$ and $q$.

\section{Logarithmic lower bound in constant diameter graphs}
\label{sec:lower-bound-topologies}
In this section we prove the following theorem:
\ThmDiameter*

The basic idea of the proof is the the same as that introduced by Göös and Suomela in their proof for the $\Omega(\log n)$ lower bound on certificate size for spanning trees and leader election~\cite{GoosS16} (Theorem 5.1).
They considered graph properties on cycles.
The problem is that cycles have large diameter relative to the number of vertices. 
The only change is that we will take a clique with a path extending from one of its vertices and ending in another one or a \emph{connection} of such graphs.
This way we can make the yes-instance any size necessary, without increasing the diameter.

We will connect two graphs, which will result in no-instance.
By \emph{connection} we mean taking two graphs, removing the edge in the middle of the path in each graph, and adding two edges to connect the corresponding endpoints of the two paths.

Then, we will assume that there is a proof labeling scheme $(f,\mathcal{A})$ of size $o(\log(n))$.
Finally, we will take the graph created by connecting yes-instances, retaining the identifiers and certificates of the original graphs. 
This way, the verifier will have to accept the combined graph even though it will be a no-instance.  

In the context of the AMOS problem, consider several graphs, each with one selected vertex.
By connecting two graphs, we obtain a graph with two selected vertices.
The verifier will also have to accept this connected graph.

\begin{proof}
    Consider a family of graphs $\mathcal{F}$, which includes graphs that have diameter at most $d$ and 
    consist of a cycle, where any vertex $v$ can be replaced with a clique $K$, such that the two edges incident to $v$ are connected to any two vertices of $K$.
    It is possible to have a constant number of bits of auxiliary information on vertices, such as labels or colors.
    Fix a radius $r$ and a diameter $d$. 
    We assume that $d$ is much larger than $r$ but at least $5r$.
    Let $\mathcal{P} \subseteq \mathcal{F}$ be a graph property such that every graph in $G\in\mathcal{P}$ has diameter at most $d/2$ and one selected vertex.
    Assume there is an $r$-local proof labeling scheme $(f,\mathcal{A})$ of size $o(log(n))$.

    We show that there are always yes-instances that can be connected to form an instance that inherits the certificates from the yes-instances.
    The verifier will accept each yes-instance, and therefore it will also accept the connected instance.

\begin{figure}[!h]
  \centering
  \includegraphics[width=0.7\textwidth]{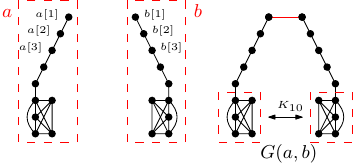}
  \caption{An example of a graph $G(a,b)$ consisting of $a$ and $b$.}
  \label{fig:yesinstance}
\end{figure}

    Similarly to the proof in~\cite{GoosS16}, we split $n = n_A + n_B$, where $n_A = \lfloor n/2\rfloor$ and $n_B = \lfloor n/2\rfloor$.
    Fix a partition of the set of identifiers $\{1,...,n^2\}$ into $2n$ subsets $a_1,...,a_n,b_1,...,b_n$ so that each $a \in A := \{a_1,...,a_n\}$ is of size $n_A$, and each $b \in B := \{b_1,...,b_n\}$ is of size $n_B$.
    Let $a[i]$ (resp., $b[i]$) denote the $i$-th identifier in $a$ (resp., $b$) in the natural order.

    We define a family of yes-instances $G(a,b)$ indexed by pairs $(a,b) \in A \times B$.
    The graph $G(a,b)$ consists of:
    \begin{itemize}
        \item a path $a[1],a[2],\dots,a[2r+1],\dots,a[k]$,
        \item a path $b[1],b[2],\dots,b[2r+1],\dots,b[\ell]$,     
        \item a clique consisting of the vertices $a[k+1],\dots,a[n_A]$ and $b[\ell+1],\dots,b[n_B]$, 
        \item and the edges $\{a[1],b[1]\},\{a[k],a[k+1]\},\{b[\ell],b[\ell+1]\}$.
    \end{itemize}
    For an example, see \Cref{fig:yesinstance}.
    
    Recall that we assume that $\mathcal{P}$ requires only an $o(\log(n))$-bits certificate to be verified.
    Let $c(a,b)$ denote a color of $G(a,b)$, consisting of all the labels and certificates within distance $2r+1$ of the vertices $a[1]$ and $b[1]$.
    Formally $c(a,b)$ is defined as follows:

\[
c(a,b)=\left(\begin{aligned}
& f(G(a,b))(a[2r+1]),\ L(G(a,b))(a[2r+1]),\ \dots,\ f(G(a,b))(a[1]),\ L(G(a,b))(a[1]),\\
& f(G(a,b))(b[1]),\ L(G(a,b))(b[1]),\ \dots,\ f(G(a,b))(b[2r+1]), \ L(G(a,b))(b[2r+1])
\end{aligned}\right)
\]

It holds that the number of bits of $c(a,b)$ is $o(r\log(n))$.

Let $K_{n,n} = (A \cup B,E)$ be a complete bipartite graph with $E =\{\{a,b\} : a \in A,b \in B\}$. 
We define an edge coloring of $K_{n,n}$ as follows: 
the color of the edge $\{a,b\} \in E$ is $c(a,b)$.
For sufficiently large $n$, the number of bits in $c(a,b)$ is smaller than $\log(n)/3$.
Therefore, the number of distinct colors in $K_{n,n}$ is therefore smaller than $\sqrt[3]{n}$. 
Hence there is a subset of edges $H \subseteq E$ such that $|H| > |E|/\sqrt[3]{n} = n^{5/3}$ and all edges of $H$ have the same color.
In the original proof, the authors apply the result due to Bondy and Simonovits \cite{BONDY197497} to show that in $H$ there is a $2k$-cycle where all the edges have the same color. 
We only need a weaker result, namely \emph{Kővári–Sós–Turán theorem} \cite{Kovari1954}, from which it follows that for a sufficiently large $n$, the subgraph $(A\cup B,H)$ necessarily contains a $4$-cycle. 
Let the vertices of this cycle be $a_i,b_i,a_j,b_j$ in this order, where $a_i,a_j \in A$ and $b_i,b_j \in B$. 
As all edges of the cycle have the same color, we have $c(a_i,b_i) = c(a_j,b_j) = c(a_i,b_j) = c(a_j,b_i)$.

\begin{figure}[!h]
  \centering
  \includegraphics[width=0.6\textwidth]{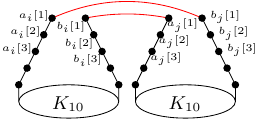}
  \caption{An example of a graph created by connecting two yes-instances $G(a_i,b_i)$ and $G(a_j,b_j)$. }
  \label{fig:noinstance}
\end{figure}

Now we connect the instances $G(a_i,b_i)$ and $G(a_j,b_j)$ and construct a no-instance.
See \Cref{fig:noinstance} for an example.

That is, we take the vertex-disjoint graphs $G(a_i,b_i)$ and $G(a_j,b_j)$, remove the edges $\{a_i[1],b_i[1]\}$ and $\{a_j[1],b_j[1]\}$, and add $\{b_{i}[1],a_j[1]\}$ and $\{b_j[1],a_i[1]\}$.
For each vertex $v$, we keep the auxiliary information $L(v)$ and the
certificates $P(v)$ from the graphs $G(a_i,b_i)$ and $G(a_j,b_j)$.
Note that since we connect instances with diameter at most $d/2$ the diameter of the resulting instance is at most $d$ and thus it still belongs to $\mathcal{F}$. 

It remains to argue that the computation of $\mathcal{A}$ on the connected instance $G$, with the labels $L$ and the certificate assignment $P$, is accepting.
To see this, pick a vertex $v$.
Then there exists $i$ such that $v \in V(G(a_i
,b_i))$. 
If $v$ is far from $a_i[1]$ and $b_i[1]$, then the local neighborhood of $v$ looks identical in $G$ and $G(a_i,b_i)$. 
Because $G(a_i,b_i)$ is a yes-instance, $\mathcal{A}(v)$ accepts.
If $v$ is close to $b_i[1]$, then the local neighborhood of $v$ looks identical in $G$ and $G(a_{i+1},b_i)$, which is another yes-instance. Similarly, if $v$ is near $a_i[1]$, then its local neighborhood looks identical in $G$ and $G(a_i,b_{i-1})$, which is also a yes-instance.
In all cases, $\mathcal{A}(v)$ accepts.
Thus, the connected graph is accepted by all vertices.
If $G \not\in \mathcal{P}$, this have a contradiction.
We can therefore conclude that the graph property $\mathcal{P}$ does not admit proof labeling scheme of size $o(\log n)$.

For the AMOS problem, place one selected vertex in each yes-instance.
When these yes-instances are combined, the resulting graph contains more than one selected vertex.
Thus, the graph with more than one selected vertex is accepted.

\end{proof}

\begin{observation} 
Certifying AMOS in bipartite graphs, and thus also in perfect graphs, with bounded diameter, requires certificates of size $\Omega(\log n)$.
\end{observation}

\begin{proof}
    We can achieve the same result for bipartite graphs by subdividing each edge once, that is, by inserting a vertex into each edge.
    Since all bipartite graphs are perfect, the lower bound also applies to perfect graphs.
\end{proof}

\begin{observation}
Certifying AMOS in graphs with bounded neighborhood diversity requires certificates of size $\Omega(\log n)$.
\end{observation}

\begin{proof}
Our yes-instances also have small neighborhood diversity, since each vertex on the path can form its own set, and the clique forms a single set. 
Therefore, the lower bound $\Omega(\log n)$ also holds for the smaller class of graphs with bounded neighborhood diversity.
\end{proof}

\section{Log-diameter certification in anonymous chordal graphs}

Let $\mathcal{F}$ be the family of all labeled chordal graphs with diameter at most $D$, and let $\mathcal{P}\subseteq\mathcal{F}$ denote the set of graphs satisfying the AMOS property, meaning that each $G\in\mathcal{P}$ contains at most one selected vertex.In this section we will prove the following theorem:

\ThmChordal*

Let us first describe the scheme, namely the prover strategy on yes-instances and the verification at the vertices.
    
    Let $G$ be a yes-instance and $P = f(G)$ be the certificate assignment for $G$. 
    Let $s$ be the selected vertex. 
    If no vertex is selected, then $s$ is chosen arbitrarily. 
    The certificate of $s$ is $P(s) = 0$.
     
    For any other $v\in V$, the certificate assignment is $P(v) = d(v,s)$. 

    Verification on a vertex $v$ consists of the following steps:
    \begin{itemize}
        \item If $v$ is selected and $P(v) \neq 0$, the verifier $\mathcal{A}(v)$ rejects. 
        \item If $P(v) = 0$ and there exists $y \in N(v)$ such that $P(y) \neq 1$, the verifier $\mathcal{A}(v)$ rejects.
        \item If $P(v)>0$, then:

        \begin{enumerate}
            \item[R1:] If for all $u\in N(v)$ we have $P(u)\geq P(v)$, or if there exists $u\in N(v)$ such that $|P(u) - P(v)| >~1$, the verifier $\mathcal{A}(v)$ rejects.
            \item[R2:] If there exist vertices $u_1,u_2 \in N(v)$ such that $P(u_1) = P(u_2) = P(v) -1$ 
            and there is no edge $\{u_1,u_2\}$ then the verifier $\mathcal{A}(v)$ rejects.
            \item[R3:] If there exists a vertex $u \in N(v)$ with $P(u) = P(v) -1$ and a vertex $y\in N(v)$ with $P(y) = P(v)$, then for each such $y$, if
            \begin{enumerate}
                \item $u\notin N(y)$, and
                \item there exists $z\in N(y)$ such that $P(z) = P(v)-1$ and $z\notin N(v)$
            \end{enumerate} 
            the verifier $\mathcal{A}(v)$ rejects.
        \end{enumerate}
        \item Otherwise, the verifier $\mathcal{A}(v)$ accepts.
    \end{itemize}

If the condition in R2 is not satisfied, the input graph is either not chordal, as there exists a cycle $v,u_1, \dots, u_2$ of length greater than 3 or there are two vertices $s_1,s_2$.
See \Cref{fig:chordalR2} for a visualization of this case.

For the case of R3, see \Cref{fig:chordalR3}.

\begin{figure}[!h]
  \centering
  \includegraphics[width=0.6\textwidth]{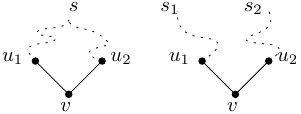}
  \caption{Visualization of a case that R2 prevents.}
  \label{fig:chordalR2}
\end{figure}

\begin{figure}[!h]
  \centering
  \includegraphics[width=0.6\textwidth]{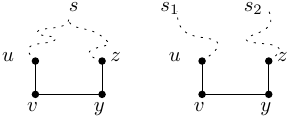}
  \caption{Visualization of a case that R3 prevents. Note that even if there would be an edge between $u$ and $z$ the cycle would still be too long.}
  \label{fig:chordalR3}
\end{figure}


     


    \begin{proof}[Proof of Theorem~\ref{thm:chordal}]
     First, we show that if a graph is accepted, it must contain at most one selected vertex.
     \newline
        Proof of $ \Rightarrow: $ 
        
        We can assume that there is a vertex $s$ with $P(s) = 0$. If not, by R1 the verifier would reject the vertex with the smallest certificate, as it would have no neighbor with a smaller certificate.
        
        The proof proceeds by induction on the distance from the vertex~$s$. 
        We assume that for any vertex~$v$ with $d(v,s) < k$, we have $P(v) = d(v,s)$.

       Base case: Every vertex $u \in G[s,r]$ can verify that its certificate matches its actual distance from~$s$.
        
        Induction step:         
        We show that, under the induction hypothesis, for any vertex $v$ with $d(v,s) = k$, either $P(v) = d(v,s)$ or the verifier rejects at some vertex.
        Consider a vertex $v$ with $d(v,s) = k$ and any neighbor $u \in N(v)$ with $d(u,s) = k-1$.
         
        Let $w \in N(u)$ be the neighbor of $u$ such that $d(w,s) =k-2$.
            
        \begin{itemize}
            \item If $P(v) < k-2$ or $P(v) > k$, the verifier $\mathcal{A}(v)$ rejects according to R1.
            
            \item If $P(v) = k-2$, recall that vertex $u$ sees its neighbor $w$ with $P(w) = k-2$.
            However, there cannot be an edge ${v,u}$; if such an edge existed, it would imply that $d(v,s) \leq k-1$, which is not the case. 
            Therefore, the verifier $\mathcal{A}(u)$ rejects according to R2.
            
            \item Now suppose that $P(v) = k-1$. 
             Vertex $v$ must have a neighbor $y$ with $P(y) = k-2$; otherwise, the verifier $\mathcal{A}(v)$ would reject according to R1.
            However, $y$ cannot have $d(y,s) = k-2$, as that would imply $d(v,s) \ne k$. 
            This means that $P(y) \ne d(y,s)$ and $y$ is attempting to \emph{deceive} the verifier. 
            Since $y \in N(v)$ and $d(v,s) = k$, there are only three possible values for $d(y,s)$: 
            \begin{enumerate}
                \item If $d(y,s) = k-1$, we reach a contradiction with the induction hypothesis, which states that the certificates are correct for all vertices at distance less than $k$ from $s$. 
                
                \item If $d(y,s) = k$, then $y$ cannot be adjacent to any vertex at distance $k-2$, as that would imply $d(y,s) \leq k-1$. 
                In particular, $y$ cannot be adjacent to $w$.
                Now, if there is an edge between $u$ and $y$, the verifier $\mathcal{A}(u)$ rejects according to R2, because $P(w) = P(y) = k-2$, $w, y \in N(u)$, but $\{w,y\} \notin E$.
                If there is no edge $\{u,y\}$, then both $\mathcal{A}(u)$ and $\mathcal{A}(v)$ reject due to R3, reason being that the vertex $v$ cannot be adjacent to $w$ and $y$ is not adjacent to~$u$.
                
                \item If $d(y,s) = k+1$, then $u$ cannot be adjacent to $y$, as that would imply $d(y,s) \leq k$.
                Without the edge $\{u,y\}$, the verifier $\mathcal{A}(v)$ rejects according to R2.

            \end{enumerate} 
            All possibilities for the distance of such a \emph{deceiving} vertex $y$ have been exhausted, each leading to a contradiction.
            Therefore, no such vertex $y$ exists, and consequently no such vertex $v$ exists either.
            
            \item The only remaining option is that $P(v) = k = d(v,s)$, which is the correct certificate.
        
        \end{itemize}
        For the converse implication, we show that any yes-instance is always accepted. 
        \newline
        Proof of $\Leftarrow:$

        If there is exactly one selected vertex $s$, then for every vertex $v$, the certificate is $P(v) = d(v,s)$, and it holds that $P(s) = 0$.
        If there is no selected vertex in $G$, then any vertex can serve as $s$, and the certificate $P(s) = 0$ is assigned accordingly. 

        R1 is always satisfied, since the certificate corresponds to the distance from $s$.
        
        R2 is satisfied because, if two vertices $x, y$ with $d(x,s) = d(y,s)$ have a common neighbor $v$ with $d(v,s) = d(x,s)+1$, there must be an edge $\{x,y\}$.
        If such an edge does not exist, then the vertices $x, v, y$ would form part of a cycle of length at least 4, as there exist paths from both $x$ and $y$ to $s$, which contradicts the assumption that the graph is chordal.
        
        R3 is satisfied for a reason similar to that of R2.
        Consider two vertices $v, y$ with $d(v,s) = d(y,s)$, and two other vertices $u \in N(v)$ and $z \in N(y)$, such that $d(u,s) = d(z,s) = d(v,s) + 1$.
        
        If neither of the edges $\{v,z\}$ or $\{y,u\}$ exists, then the vertices $u, v, y, z$ would form part of a cycle of length at least 4, similarly to the case of R2.
        Hence, the graph would not be chordal.

       Therefore, all rules are satisfied, and the graph is accepted.
    
       Since only the distances are stored, the certificates have size $O(\log(D))$.
       The verifier only checks vertices up to distance 2, so the it is 2-local.

         We have thus shown that the 2-local proof-labeling scheme $(f,\mathcal{A})$ of size $O(\log(D))$ accepts a chordal graph if and only if it contains at most one selected vertex.
    \end{proof}

\section{Log-diameter certification in anonymous grid graphs}

Let $\mathcal{F}$ be the family of all labeled grid graphs with diameter at most $D$, and let $\mathcal{P}\subseteq\mathcal{F}$ denote the set of graphs satisfying the AMOS property, meaning that each $G\in\mathcal{P}$ contains at most one selected vertex.

In this section, we prove the following:

\ThmGrid*

Again, let us first describe the prover strategy on yes-instances and then the verification at the vertices.

Let $s$ be the selected vertex and $P(s) = 0$ the certificate of $s$.
If no vertex is selected, an arbitrary vertex is chosen as $s$.
For any other vertex $v$ the certificate is $P(v)=d(v,s)$.

Verification on vertex $v$: 
 
\begin{enumerate}
    \item If $v$ is selected and $P(v)\neq 0$, then $\mathcal{A}(v)$ rejects.
    \item If $P(v) = 0$ and there is a vertex $y\in N(v)$ such that $P(y) \neq 1$, then $\mathcal{A}(v)$ rejects.  
    \item If for any $v$ there are two or more vertices $w \in G[v,r]$ such that $P(w) = 0$, then $\mathcal{A}(v)$ rejects.
    \item If for any $v$ the certificate $P(v) = k$ and there is a vertex $y\in N(v)$ such that $P(y) \notin \{k-1, k+1\}$, then $\mathcal{A}(v)$ rejects.
    \item If for any $v$ the certificate $P(v) = k>0$ and there is no $y\in N(v)$ such that $P(y) = k-1$, then $\mathcal{A}(v)$ rejects.
    \item If for any $v$ there are more than two vertices $y\in N(v)$ such that $P(y) = k-1$, then $\mathcal{A}(v)$ rejects. 
    \item Otherwise $\mathcal{A}(v)$ accepts.
    

\end{enumerate}

\begin{proof}
The proof of the left to right implication proceeds by induction on the distance from vertex $s$.
In the case there are several vertices $w$ with $P(w) = 0$, one of them is selected as $s$.

We can assume that there is such a vertex $s$ with $P(s) = 0$, otherwise, on the vertex $w$ with the smallest $P(w)$, the verifier $\mathcal{A}(w)$ would reject according to the condition 5, as there would be no vertex in $N(w)$ with a smaller certificate.

We assume that for any vertex $v$ with $d(v,s)<k$, we have $P(v) = d(v,s)$.

    Base case:
    If $s$ is a selected vertex, the certificate $P(s) = 0$; otherwise, $\mathcal{A}(s)$ would reject according to condition 1, and every $v \in N(s)$ has $P(v) = 1$ according to condition 2.

    If there are multiple vertices $w$ with $P(w) = 0$ in the local neighborhood of a vertex, the verifier rejects at that vertex.

    Induction step:
    We show that under the induction hypothesis for a vertex $v$ with $d(v,s) = k$ either $P(v) = d(v,s) = k$, or there is a vertex $y$ such that $\mathcal{A}(y)$ rejects.

    \begin{figure}[!h]
  \centering
  \includegraphics[width=0.55\textwidth]{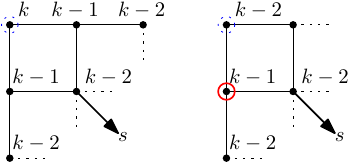}
  \caption{Visualization of the the case, where $v$ has both coordinates different from $s$. The vertex in the dotted circle is $v$. The vertex in the red circle is the vertex, where the verifier rejects.}
  \label{fig:grid1}
\end{figure}
    
    There are two cases, which will be solved separately:
    \begin{enumerate}
        \item Assume that vertex $v$ has both coordinates different from those of $s$
        In this case $v$ has two neighbors $y_1,y_2\in N(v)$ such that $P(y_1) = P(y_2) = k-1$. 
        The certificate $P(v)$ can be equal to $k$, in which case it is the correct distance, or $P(v) = k-2$. 
        In any other case, the vertices $y_1, y_2$ would reject due to condition 4.
        
        Now assume that $P(v) = k-2$.
        Since $v$ differs in both coordinates, either $y_1$ or $y_2$ has to differ in both coordinates too. 
        If not, they would be neighboring to $s$ and they would reject due to condition $3$ as $k-1 = 0$. 
        Without loss of generality, let $y_1$ be the vertex that differs in both coordinates.
        
        The vertex $y_1$ has to have two neighbors at distance $k-2$ from $s$ equal to $k-2$ and $v$ would be a third neighbor with $P(v) = k-2$ and thus $\mathcal{A}(y_1)$ would reject due to condition $6$.

        See \Cref{fig:grid1} for a visualization of this case. 
        
        Thus the only option is that $P(v) = k = d(v,s)$, which is the correct certificate.
        
        \item Now assume that vertex $v$ differs from $s$ in only one coordinate in the grid.
        Let $y\in N(v)$ be the vertex with $d(y,s) = P(y) = k-1$.

        The certificate $P(v)$ can, again, be either $k$ or $k-2$ otherwise $\mathcal{A}(y)$ would reject.
        If $P(v) = k$ the certificate is correct, so let us assume that $P(v) = k-2$ and vertex $v$ is trying to deceive the verifier.

        Let $v_N, v_E, v_S, v_W$ denote the neighbors of $y$, and let $v = v_N$ and $P(v_S) = k-2$. 
        
        Let us emphasize that the vertices $v_S, v_E, v_W$ must have a correct certificate $P(v_i) = d(v_i,s)$ for $i\in \{S,E,W\}$. 
        For the vertex $v_S$, this holds from the induction hypothesis since $d(s,v_S) < k$, and $v_E,v_W$ can have the certificate equal to either $k-2$ or $k$.
        However, if $P(v_E) = k-2$ or $P(v_W) = k-2$ there would be $3$ neighbors of $y$ with certificate $k-2$ and thus $\mathcal{A}(y)$ would reject. 
        The only possibility is that $v_S, v_E, v_W$ have the correct certificate.

        The vertices $v$ and $v_E$ have a common neighbor $v_C$. 
        As $P(v) = k-2$ and $P(v_E) = k$ the only possible certificate is $P(v_C) = k-1$. 
        But then $v_E$ has three neighbors whose certificates are equal to $k-1$ and $\mathcal{A}(v_E)$ rejects.  

    See \Cref{fig:grid2} for a visualization of this case.

        Thus we showed that the only possibility is $P(v) = k$ which is the correct certificate.
                
    \end{enumerate}

    And thus we have shown that if our proof labeling scheme accepts, the AMOS property holds.
     For the converse implication, see \Cref{app:proof_grids}.    

\end{proof}

\begin{figure}[!h]
  \centering
  \includegraphics[width=0.6\textwidth]{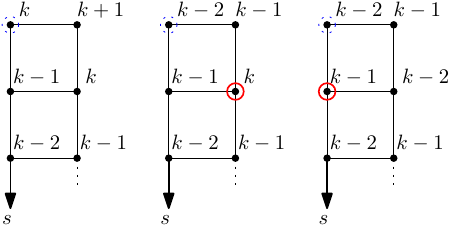}
  \caption{Visualisation of the the case, where $v$ has one coordinate different from $s$. The vertex in the dotted circle is $v$. The vertex in the red circle is the vertex, where the verifier rejects.}
  \label{fig:grid2}
\end{figure}

\section{Sublogarithmic upper bounds in everywhere-dense graphs}
\label{sec:everywhere-dense}

Up to this point, we have explored two aspects of the topology of the lower bound: the diameter and the presence of a hole. We now turn our attention to the density of the graph. For the Göös-Suomela technique to work, it is essential to be able to cut \emph{yes}-instances into pieces and combine these pieces without too many vertices noticing the difference. If this cannot be done, then the counting argument fails. We prove that if the graph is dense everywhere, in the crude sense that the degree is relatively high for all vertices, then the lower bound cannot work, since we have a sublogarithmic upper bound. 

\ThmDenseEverywhere*

Due to the lack of space, the proof of this theorem is deferred to Appendix~\ref{app:density}. Let us briefly sketch the technique. For AMOS, we would like to give the identifier of the selected vertex (if it exists) to all vertices, but for this we would need $\Theta(\log n)$ bits. The idea is to cut this identifier into small pieces and distribute them to the vertices. Then a vertex will use one round of communication to gather all the pieces and recover the leader identifier, and one round to check that it has been given the same as its neighbors, and that if it is selected then this is its own identifier. Arguing that there exists a good distribution of small enough pieces can be done via the probabilistic method, in a similar way to \cite{FeuilloleyFHPP21, Cook0M25}.

\section{Sublogarithmic upper bound with small identifiers}
\label{sec:small-ID}
In this section, we prove the following theorem.

\ThmIDs*

The idea of the proof is to choose the vertex with the smallest identifier as the root of a spanning tree, and then accumulate the number of selected vertices from the leaves up to the root.

For the full proof, see \Cref{app:proof_ids}.

\begin{corollary}
If the identifiers are in $[1+k, n+h]$ for some constant integers $h$ and $k$, with $k \leq h$, then AMOS on a graph with a bounded diameter $D$ can be done with certificates of size $O(\log D)$.
\end{corollary}
The corollary follows directly from the theorem.

\newpage{}

\bibliographystyle{plain}
\bibliography{biblio-GS-bound}

\newpage{}

\appendix
\section{Proof of right-to-left implication of Theorem \ref{thm:grids}}\label{app:proof_grids}
In this appendix, we show the right-to-left implication of \Cref{thm:grids}.

\ThmGrid*

\begin{proof}
    We show that any yes-instance is always accepted.
    
    Assume that there is one selected vertex $s$ and that for every $v$ the certificate $P(v) = d(v,s)$. 
    If there is no selected vertex, an arbitrary vertex is chosen instead. 
    
    Because $P(v) = d(v,s)$, the conditions in points 1, 2, and 4 cannot be true. 
    Since there is at most one selected vertex, the condition in point $3$ will also never be true.
    The condition in point 5 is also never true, because the certificate is a distance from a specific vertex. If there is a vertex with distance $k$, there needs to be a neighbor with distance $k-1$.
    Lastly, the condition in point 6 cannot be true, since in a grid it is not possible for three different neighbors to be closer to a given vertex.

    We have thus shown both implications, and the theorem is proved.
    
\end{proof}

\section{Proof of Theorem \ref{thm:small_ids}}\label{app:proof_ids}

We give the full proof of the following theorem:

\ThmIDs*

\begin{proof}
    The vertex with the smallest identifier can be selected as the root of a spanning tree.
    Each vertex's certificate encodes the root's identifier and the distance from the root.
    In addition, there is an integer $x$ that represents the number of selected vertices in the subtree rooted at $v$.
 
     Formally, for each vertex $v$, the certificate $P(v) = (id,d,s)$, where $id,d,s$ are non-negative integers.
     Let $id_v,d_v$ and $s_v$ be the respective elements of $P(v)$ of a vertex $v$.

    Verification on vertex $v$ is done as follows:
    \begin{itemize}
        \item If there is a neighbor $u\in N(v)$ such that $id_v\neq id_u$, the verifier $\mathcal{A}(v)$ rejects.
        \item If $d_v = 0$ and $v \neq id_v$, $\mathcal{A}(v)$ rejects.
        \item If $v = id_v$ and $d \neq 0$, $\mathcal{A}(v)$ rejects.        
        \item If $d_v>0$ and:
        \begin{itemize}
            \item there is not exactly one neighbor $u\in N(v)$ such that $d_u = d_v-1$, $\mathcal{A}(v)$ rejects.
            \item any other neighbor $w\in N(v)\setminus u$ has $d_w \neq d_v$ and $d_w \neq d_v+1$, $\mathcal{A}(v)$ rejects.
        \end{itemize}
        \item If $v$ is not selected and $s_v\neq\sum_ws_w$ where $w\in N(v)$ such that $d_w\geq d_v$, $\mathcal{A}(v)$ rejects. 
        \item If $v$ is selected and $s_v\neq1+\sum_ws_w$ where $w\in N(v)$ such that $d_w\geq d_v$, $\mathcal{A}(v)$ rejects.
        \item Otherwise, $\mathcal{A}(v)$ accepts.
    \end{itemize}

    Correctness follows from the correctness of the spanning tree certification, see \emph{e.g.} \cite{KormanKP10}.

    As the root is the vertex with the smallest identifier, which is of constant size relative to the size of the graph, it can be encoded using a constant number of bits.
    The distance $d$ can be encoded using $\log D$ bits, as the graph has a bounded diameter $D$.
    The number of selected vertices is from the set $\{0,1\}$, if not, the certificate is rejected.
    The certificate is thus of size $\log D$.
    
\end{proof}

\section{Proof of Theorem~\ref{thm:dense-everywhere}}
\label{app:density}

Let us remind the theorem.

\ThmDenseEverywhere*

\begin{proof}
Let us describe the certification scheme. If no vertex is selected, the prover assigns a special label to all vertices, and the vertices can then verify that they are indeed not selected.
Otherwise, if there is exactly one selected vertex, the prover ensures that all vertices can reconstruct the same identifier of the selected vertex and verify whether it matches their status (selected or not).
More precisely, on a correct instance, the prover assigns to every vertex a certificate $(i,S[i])$, where $i$ is an integer in $O(\log n)$, hence encoded using $O(\log \log n)$ bits, and $S[i]$ is the $i$-th bit  of the identifier $S$ of the selected vertex. 
The verification proceeds as follows. Every vertex gathers all the pairs $(i, S[i])$ from its neighborhood, concatenates the bits in the correct order, and then checks with its neighbors that they have reconstructed the same identifier. 
If this is not the case, the vertex rejects; otherwise, it checks consistency with its input: if it is selected, the reconstructed identifier should be the same as its own identifier, and if it is not selected, it should be different. Note that this requires communication radius~2, since a node needs to know what identifiers its neighbors have reconstructed.
Note that we need a radius of at least 2 to allow one more round of communication to exchange the reconstructed identifiers. 

Clearly this protocol cannot accept a no-instance. 
The only thing left to prove is that all \emph{yes}-instances can be accepted. This boils down to checking that there is always a way for the prover to assign the pairs $(i,S[i])$ to the vertices so that, for every vertex, all the bits of $S$ are present in its neighborhood. This is where we need the lower bound on the degree, since with very small degree it would clearly be impossible. 

We prove this using the probabilistic method. Consider assigning a bit position to each vertex uniformly at random. Fix a vertex $v$ and a specific bit position $b$. Let $\beta$ be a constant such that the identifiers are encoded using at most $\beta \log n$ bits. 
The probability that $b$ does not appear in the neighborhood of $v$ is:

\[\left(1-\frac{1}{\beta \log n} \right)^{\alpha \log^2 n}
=
\left(1-\frac{(\alpha/\beta)\log n}{\alpha \log^{2} n} \right)^{\alpha \log^2 n} 
\leq 
e^{-\alpha \log n/\beta} 
\leq \frac{1}{n^3}.\]

For large enough $\alpha$ this is smaller than $1/n^3$. Now, taking a union bound over all vertices and all bits, we still get a probability that is strictly smaller than 1. 
Therefore, such an assignment where all vertices see all the bits must exist.

\end{proof}

\end{document}

%% file: hole-thin.tex
\tikzset{every picture/.style={line width=0.75pt}} 

\begin{tikzpicture}[x=0.75pt,y=0.75pt,yscale=-1,xscale=1]

\draw  [fill={rgb, 255:red, 185; green, 185; blue, 185 }  ,fill opacity=1 ] (107.97,63.73) .. controls (116.18,54.66) and (137.1,48.94) .. (159.74,47.02) .. controls (189.48,44.49) and (222.19,48.49) .. (233,60) .. controls (252.03,80.27) and (249.97,117.73) .. (236.97,135.73) .. controls (223.97,153.73) and (120.97,154.73) .. (104.97,132.73) .. controls (94.5,118.34) and (90.88,102.24) .. (94.95,86.94) .. controls (97.1,78.86) and (101.4,71) .. (107.97,63.73) -- cycle ;
\draw  [fill={rgb, 255:red, 255; green, 255; blue, 255 }  ,fill opacity=1 ] (134.97,68.73) .. controls (143.97,56.73) and (146.84,50.88) .. (169.9,51.31) .. controls (192.97,51.73) and (191.97,50.73) .. (203.97,62.73) .. controls (215.97,74.73) and (218.97,121.73) .. (208.97,132.73) .. controls (198.97,143.73) and (188.02,144.77) .. (170.95,144.68) .. controls (153.89,144.59) and (141.57,141.82) .. (134.97,132.73) .. controls (131.07,127.38) and (128.64,110.28) .. (129.15,94.67) .. controls (129.5,83.8) and (131.27,73.66) .. (134.97,68.73) -- cycle ;
\draw  [dash pattern={on 4.5pt off 4.5pt}] (147.9,41.31) -- (191.9,41.31) -- (191.9,61.31) -- (147.9,61.31) -- cycle ;
\draw  [dash pattern={on 4.5pt off 4.5pt}] (149.9,137.31) -- (193.9,137.31) -- (193.9,157.31) -- (149.9,157.31) -- cycle ;

\draw (154,91.4) node [anchor=north west][inner sep=0.75pt]    {$Hole$};
\draw (153,21.4) node [anchor=north west][inner sep=0.75pt]    {$Thin$};
\draw (155,165.4) node [anchor=north west][inner sep=0.75pt]    {$Thin$};

\end{tikzpicture}

%% file: chordal.tex
\tikzset{every picture/.style={line width=0.75pt}} 

\begin{tikzpicture}[x=0.75pt,y=0.75pt,yscale=-1,xscale=1]

\draw  [fill={rgb, 255:red, 0; green, 0; blue, 0 }  ,fill opacity=1 ] (104,119.98) .. controls (104,116.68) and (106.68,114) .. (109.98,114) .. controls (113.29,114) and (115.97,116.68) .. (115.97,119.98) .. controls (115.97,123.29) and (113.29,125.97) .. (109.98,125.97) .. controls (106.68,125.97) and (104,123.29) .. (104,119.98) -- cycle ;
\draw  [fill={rgb, 255:red, 0; green, 0; blue, 0 }  ,fill opacity=1 ] (105,166.98) .. controls (105,163.68) and (107.68,161) .. (110.98,161) .. controls (114.29,161) and (116.97,163.68) .. (116.97,166.98) .. controls (116.97,170.29) and (114.29,172.97) .. (110.98,172.97) .. controls (107.68,172.97) and (105,170.29) .. (105,166.98) -- cycle ;
\draw  [fill={rgb, 255:red, 0; green, 0; blue, 0 }  ,fill opacity=1 ] (151,119.98) .. controls (151,116.68) and (153.68,114) .. (156.98,114) .. controls (160.29,114) and (162.97,116.68) .. (162.97,119.98) .. controls (162.97,123.29) and (160.29,125.97) .. (156.98,125.97) .. controls (153.68,125.97) and (151,123.29) .. (151,119.98) -- cycle ;
\draw  [fill={rgb, 255:red, 0; green, 0; blue, 0 }  ,fill opacity=1 ] (152,166.98) .. controls (152,163.68) and (154.68,161) .. (157.98,161) .. controls (161.29,161) and (163.97,163.68) .. (163.97,166.98) .. controls (163.97,170.29) and (161.29,172.97) .. (157.98,172.97) .. controls (154.68,172.97) and (152,170.29) .. (152,166.98) -- cycle ;
\draw  [fill={rgb, 255:red, 0; green, 0; blue, 0 }  ,fill opacity=1 ] (103,69.98) .. controls (103,66.68) and (105.68,64) .. (108.98,64) .. controls (112.29,64) and (114.97,66.68) .. (114.97,69.98) .. controls (114.97,73.29) and (112.29,75.97) .. (108.98,75.97) .. controls (105.68,75.97) and (103,73.29) .. (103,69.98) -- cycle ;
\draw  [fill={rgb, 255:red, 0; green, 0; blue, 0 }  ,fill opacity=1 ] (166,67.98) .. controls (166,64.68) and (168.68,62) .. (171.98,62) .. controls (175.29,62) and (177.97,64.68) .. (177.97,67.98) .. controls (177.97,71.29) and (175.29,73.97) .. (171.98,73.97) .. controls (168.68,73.97) and (166,71.29) .. (166,67.98) -- cycle ;
\draw  [fill={rgb, 255:red, 0; green, 0; blue, 0 }  ,fill opacity=1 ] (203,108.98) .. controls (203,105.68) and (205.68,103) .. (208.98,103) .. controls (212.29,103) and (214.97,105.68) .. (214.97,108.98) .. controls (214.97,112.29) and (212.29,114.97) .. (208.98,114.97) .. controls (205.68,114.97) and (203,112.29) .. (203,108.98) -- cycle ;
\draw  [fill={rgb, 255:red, 0; green, 0; blue, 0 }  ,fill opacity=1 ] (214,63.98) .. controls (214,60.68) and (216.68,58) .. (219.98,58) .. controls (223.29,58) and (225.97,60.68) .. (225.97,63.98) .. controls (225.97,67.29) and (223.29,69.97) .. (219.98,69.97) .. controls (216.68,69.97) and (214,67.29) .. (214,63.98) -- cycle ;
\draw  [fill={rgb, 255:red, 0; green, 0; blue, 0 }  ,fill opacity=1 ] (198,167.98) .. controls (198,164.68) and (200.68,162) .. (203.98,162) .. controls (207.29,162) and (209.97,164.68) .. (209.97,167.98) .. controls (209.97,171.29) and (207.29,173.97) .. (203.98,173.97) .. controls (200.68,173.97) and (198,171.29) .. (198,167.98) -- cycle ;
\draw [line width=1.5]    (108.98,69.98) -- (109.98,119.98) ;
\draw [line width=1.5]    (109.98,119.98) -- (110.98,166.98) ;
\draw [line width=1.5]    (156.98,119.98) -- (157.98,166.98) ;
\draw [line width=1.5]    (171.98,67.98) -- (156.98,119.98) ;
\draw [line width=1.5]    (208.98,108.98) -- (156.98,119.98) ;
\draw [line width=1.5]    (219.98,63.98) -- (171.98,67.98) ;
\draw [line width=1.5]    (219.98,63.98) -- (208.98,108.98) ;
\draw [line width=1.5]    (156.98,119.98) -- (110.98,166.98) ;
\draw [line width=1.5]    (157.98,166.98) -- (109.98,119.98) ;
\draw [line width=1.5]    (203.98,167.98) -- (156.98,119.98) ;
\draw    (203.98,167.98) -- (157.98,166.98) ;
\draw [line width=1.5]    (156.98,119.98) -- (109.98,119.98) ;
\draw [line width=1.5]    (157.98,166.98) -- (110.98,166.98) ;
\draw [line width=1.5]    (171.98,67.98) -- (208.98,108.98) ;

\end{tikzpicture}

%% file: grid.tex
\tikzset{every picture/.style={line width=0.75pt}} 

\begin{tikzpicture}[x=0.75pt,y=0.75pt,yscale=-1,xscale=1]

\draw  [fill={rgb, 255:red, 0; green, 0; blue, 0 }  ,fill opacity=1 ] (104,119.98) .. controls (104,116.68) and (106.68,114) .. (109.98,114) .. controls (113.29,114) and (115.97,116.68) .. (115.97,119.98) .. controls (115.97,123.29) and (113.29,125.97) .. (109.98,125.97) .. controls (106.68,125.97) and (104,123.29) .. (104,119.98) -- cycle ;
\draw  [fill={rgb, 255:red, 0; green, 0; blue, 0 }  ,fill opacity=1 ] (105,166.98) .. controls (105,163.68) and (107.68,161) .. (110.98,161) .. controls (114.29,161) and (116.97,163.68) .. (116.97,166.98) .. controls (116.97,170.29) and (114.29,172.97) .. (110.98,172.97) .. controls (107.68,172.97) and (105,170.29) .. (105,166.98) -- cycle ;
\draw  [fill={rgb, 255:red, 0; green, 0; blue, 0 }  ,fill opacity=1 ] (151,119.98) .. controls (151,116.68) and (153.68,114) .. (156.98,114) .. controls (160.29,114) and (162.97,116.68) .. (162.97,119.98) .. controls (162.97,123.29) and (160.29,125.97) .. (156.98,125.97) .. controls (153.68,125.97) and (151,123.29) .. (151,119.98) -- cycle ;
\draw  [fill={rgb, 255:red, 0; green, 0; blue, 0 }  ,fill opacity=1 ] (152,166.98) .. controls (152,163.68) and (154.68,161) .. (157.98,161) .. controls (161.29,161) and (163.97,163.68) .. (163.97,166.98) .. controls (163.97,170.29) and (161.29,172.97) .. (157.98,172.97) .. controls (154.68,172.97) and (152,170.29) .. (152,166.98) -- cycle ;
\draw [line width=1.5]    (109.98,119.98) -- (110.98,166.98) ;
\draw [line width=1.5]    (156.98,119.98) -- (157.98,166.98) ;
\draw [line width=1.5]    (156.98,119.98) -- (137.3,119.98) -- (109.98,119.98) ;
\draw [line width=1.5]    (157.98,166.98) -- (110.98,166.98) ;
\draw  [fill={rgb, 255:red, 0; green, 0; blue, 0 }  ,fill opacity=1 ] (200,119.98) .. controls (200,116.68) and (202.68,114) .. (205.98,114) .. controls (209.29,114) and (211.97,116.68) .. (211.97,119.98) .. controls (211.97,123.29) and (209.29,125.97) .. (205.98,125.97) .. controls (202.68,125.97) and (200,123.29) .. (200,119.98) -- cycle ;
\draw  [fill={rgb, 255:red, 0; green, 0; blue, 0 }  ,fill opacity=1 ] (201,166.98) .. controls (201,163.68) and (203.68,161) .. (206.98,161) .. controls (210.29,161) and (212.97,163.68) .. (212.97,166.98) .. controls (212.97,170.29) and (210.29,172.97) .. (206.98,172.97) .. controls (203.68,172.97) and (201,170.29) .. (201,166.98) -- cycle ;
\draw  [fill={rgb, 255:red, 0; green, 0; blue, 0 }  ,fill opacity=1 ] (247,119.98) .. controls (247,116.68) and (249.68,114) .. (252.98,114) .. controls (256.29,114) and (258.97,116.68) .. (258.97,119.98) .. controls (258.97,123.29) and (256.29,125.97) .. (252.98,125.97) .. controls (249.68,125.97) and (247,123.29) .. (247,119.98) -- cycle ;
\draw  [fill={rgb, 255:red, 0; green, 0; blue, 0 }  ,fill opacity=1 ] (248,166.98) .. controls (248,163.68) and (250.68,161) .. (253.98,161) .. controls (257.29,161) and (259.97,163.68) .. (259.97,166.98) .. controls (259.97,170.29) and (257.29,172.97) .. (253.98,172.97) .. controls (250.68,172.97) and (248,170.29) .. (248,166.98) -- cycle ;
\draw [line width=1.5]    (205.98,119.98) -- (206.98,166.98) ;
\draw [line width=1.5]    (252.98,119.98) -- (253.98,166.98) ;
\draw [line width=1.5]    (252.98,119.98) -- (205.98,119.98) ;
\draw [line width=1.5]    (253.98,166.98) -- (206.98,166.98) ;
\draw [line width=1.5]    (205.98,119.98) -- (186.3,119.98) -- (156.98,119.98) ;
\draw [line width=1.5]    (206.98,166.98) -- (187.3,166.98) -- (157.98,166.98) ;
\draw [line width=1.5]    (155.98,72.98) -- (156.98,119.98) ;
\draw  [fill={rgb, 255:red, 0; green, 0; blue, 0 }  ,fill opacity=1 ] (150,72.98) .. controls (150,69.68) and (152.68,67) .. (155.98,67) .. controls (159.29,67) and (161.97,69.68) .. (161.97,72.98) .. controls (161.97,76.29) and (159.29,78.97) .. (155.98,78.97) .. controls (152.68,78.97) and (150,76.29) .. (150,72.98) -- cycle ;
\draw [line width=1.5]    (108.98,72.98) -- (109.98,119.98) ;
\draw  [fill={rgb, 255:red, 0; green, 0; blue, 0 }  ,fill opacity=1 ] (103,72.98) .. controls (103,69.68) and (105.68,67) .. (108.98,67) .. controls (112.29,67) and (114.97,69.68) .. (114.97,72.98) .. controls (114.97,76.29) and (112.29,78.97) .. (108.98,78.97) .. controls (105.68,78.97) and (103,76.29) .. (103,72.98) -- cycle ;
\draw [line width=1.5]    (155.98,72.98) -- (136.3,72.98) -- (108.98,72.98) ;
\draw [line width=1.5]    (204.98,72.98) -- (205.98,119.98) ;
\draw [line width=1.5]    (251.98,72.98) -- (252.98,119.98) ;
\draw  [fill={rgb, 255:red, 0; green, 0; blue, 0 }  ,fill opacity=1 ] (199,72.98) .. controls (199,69.68) and (201.68,67) .. (204.98,67) .. controls (208.29,67) and (210.97,69.68) .. (210.97,72.98) .. controls (210.97,76.29) and (208.29,78.97) .. (204.98,78.97) .. controls (201.68,78.97) and (199,76.29) .. (199,72.98) -- cycle ;
\draw  [fill={rgb, 255:red, 0; green, 0; blue, 0 }  ,fill opacity=1 ] (246,72.98) .. controls (246,69.68) and (248.68,67) .. (251.98,67) .. controls (255.29,67) and (257.97,69.68) .. (257.97,72.98) .. controls (257.97,76.29) and (255.29,78.97) .. (251.98,78.97) .. controls (248.68,78.97) and (246,76.29) .. (246,72.98) -- cycle ;
\draw [line width=1.5]    (251.98,72.98) -- (204.98,72.98) ;
\draw [line width=1.5]    (204.98,72.98) -- (185.3,72.98) -- (155.98,72.98) ;

\end{tikzpicture}